\documentclass[onefignum,onetabnum]{siamonline220329}

\usepackage{lipsum}
\usepackage{amsfonts}
\usepackage{graphicx}
\usepackage{epstopdf}
\usepackage{algorithmic}
\usepackage{bbm}
\usepackage{stmaryrd}
\usepackage{paralist}
\usepackage{hyperref}
\usepackage{graphicx}
\usepackage{subcaption}
\graphicspath{{figures/}}
\usepackage{color}
\usepackage{todonotes}
\usepackage{bm}
\usepackage{float}
\usepackage{setspace}
\usepackage{dsfont}
\usepackage{tabularx}
\usepackage[skip=1ex]{caption}
\usepackage{multirow}
\usepackage{wrapfig}

\ifpdf
  \DeclareGraphicsExtensions{.eps,.pdf,.png,.jpg}
\else
  \DeclareGraphicsExtensions{.eps}
\fi
\def\E{\mathbb{E}}
\def\P{\mathbb{P}}
\def\R{\mathbb{R}}

\def\one{\mathbbm{1}}

\newcommand{\barc}{{\overline{c}}}

\usepackage{enumitem}
\setlist[enumerate]{leftmargin=.5in}
\setlist[itemize]{leftmargin=.5in}

\newsiamremark{remark}{Remark}
\newsiamremark{hypothesis}{Hypothesis}
\crefname{hypothesis}{Hypothesis}{Hypotheses}
\newsiamthm{claim}{Claim}

\headers{The Price of Information }{S. Jaimungal, and X. Shi}

\title{The Price of Information
\thanks{February 17, 2024\funding{SJ would like to acknowledge support from the Natural Sciences and Engineering Research Council of Canada (RGPIN-2018-05705).}}}

\author{
Sebastian Jaimungal\thanks{Department of Statistical Sciences, University of Toronto, Toronto, Canada
  (\email{sebastian.jaimungal@utoronto.ca}, \url{sebastian.statistics.utoronto.ca}; \email{xf.shi@utoronto.ca}, \url{xf-shi.github.io}).}
\and Xiaofei Shi\footnotemark[2]}
\usepackage{amsopn}

\def\E{\mathbb{E}}
\def\P{\mathbb{P}}
\def\R{\mathbb{R}}

\def\one{\mathbbm{1}}

\newcommand{\1}[1]{{{\mathds{1}}_{\{#1\}}}}

\newcommand{\F}{{\mathcal{F}}}
\newcommand{\A}{{\mathcal{A}}}

\newcommand{\wY}{{\widehat{Y}}}
\newcommand{\wV}{{\widehat{V}}}
\newcommand{\wB}{{\widehat{B}}}
\newcommand{\wC}{{\widehat{C}}}
\newcommand{\wc}{{\widehat{c}}}

\newcommand{\T}{{\mathfrak{T}}}

\newcommand{\tT}{{t\in\T}}

\newcommand{\IA}{{\textsf{IA}}}

\ifpdf
\hypersetup{
  pdftitle={The Price of Information},
  pdfauthor={S. Jaimungal, and X. Shi}
}
\fi

\begin{document}

\maketitle

\begin{abstract}
When an investor is faced with the option to purchase additional information regarding an asset price, how much should she pay? To address this question, we solve for the indifference price of information in a setting where a trader maximizes her expected  utility of terminal wealth over a finite time horizon. If she does not purchase the information, then she solves a partial information stochastic control problem, while, if she does purchase the information, then she pays a cost and receives partial information about the asset's trajectory.  We further demonstrate that when the investor can purchase the information at any stopping time prior to the end of the trading horizon, she chooses to do so at a deterministic time.
\end{abstract}



\section{Introduction}
When stock returns are predictable, the trading signal and the anticipation of future signals affect investor's trading strategies~\cite{kim1996dynamic, merton1975optimum}.
Without access to the trading signal,  rational investors may still trade strategically by filtering the trading signal~\cite{al2018outperformance,benevs1983estimation, fouque2015filtering, guasoni2006asymmetric} -- i.e., treat the optimization problem as a partial information control problem. As  an investor with full information always has higher utility than those without the full information~\cite{ bjork2010optimal, brendle2006portfolio,brody2009informed,guasoni2006asymmetric, guasoni2019should}, possessing information about the trading signal has a strictly positive value. The investor, therefore, has a \emph{maximum} price
that she is willing to pay to acquire this additional information. 

Further, consider a group of informed traders who have access to partial information regarding the trading signal. Often there are  regulatory constraints that prevent them from legally participating in the market.  They can, however, form an information agency (\IA) to ``sell'' some aspect of the trading signal that they possess to other investors. The maximum price an investor is willing to pay, as observed in~\cite{banek2003information} using a linear-quadratic preference, may be viewed as the maximum price that the {\IA}  can sell the information to the investor. 

In this article,  we refer to this maximum price as the ``\textit{indifference price}''. In particular, the indifference price is the price that the investor is willing to pay such that their expected utility from optimally investing and filtering the trading signal (i.e., in the absence of the information) equals the expected utility from paying the indifference price, receiving the information, and then optimally trading with the information (i.e., in the presence of full information). Naturally, the investor will accept any price below the indifference price and will reject any price above the indifference price.

Our main goal is to address the  question: how does an investor quantify the price of the trading signal? We show that, on a fixed trading horizon $[0,T]$, an investor with constant absolute risk aversion (CARA) has two parameters that affect their indifference price: (i) the investor's risk aversion coefficient and (ii) the signal-to-noise ratio (or more specifically, the signal-to-volatility ratio). Further, we investigate the question of timing: when should the investor subscribe to the trading signal (until the end of the time horizon)? We show that if the investor is allowed to choose a stopping time at which to subscribe to the trading signal (and keep her subscription to the terminal time $T$), then, somewhat surprisingly, the optimal stopping time is a deterministic stopping time. 

The remainder of the paper is organized as follows. 
In Section~\ref{sec: single period}, we develop a single-period model, derive the participating investor's optimal informed and uninformed trading strategy, and obtain the closed form expression for the indifference price. 
In Section~\ref{sec: continuous}, we develop the continuous time analogue of the model, and once again determine the indifference price in closed form and derive a limiting subscription rate for subscribing to the information when the time horizon is large.
Finally, we  consider the optimal subscription time problem in Section~\ref{sec:subscribe}.
Proofs are found in Appendix~\ref{appdix: proofs}. 

Throughout, we fix a filtered probability space $(\Omega,\F, \{\F_t\}_{\tT},\P)$ supporting two independent standard Brownian motions, $(B_t^Y)_{\tT}$ and $(B_t^Z)_{\tT}$, and $\T:=[0,T]$ being the trading horizon. We use $\A(X)$ to denote the collection of all processes adapted to the filtration generated by the process $X$. 

\section{A Single-Period Model}\label{sec: single period}

We introduce the essential ideas in a single-period model.

\subsection{Single-Period Model Setup}\label{ssec: single-period setup}

 Suppose interest rates are zero, the initial price of a risky asset price is $S_0\in\R_+$, and at time $1$ the risky asset's price is 
\begin{align}
\label{dyn: single-period stock}
S_1 &= S_0 +\mu+Y + \sigma_Z B_1^Z, \qquad\mbox{where} \quad Y = y+ \sigma_Y B_1^Y. 
\end{align}
The investor knows the values of $\mu,y \in \R,\;\sigma_Y, \sigma_Z\in\R_+/\{0\}$ and they are constants. The increment of the stock price consists of three parts: (i) a (deterministic) growth rate $\mu$, (ii) a trading signal $Y$ which the investor cannot observe, and (iii) a completely random part $B_1^Z$. 

The investor begins with some initial wealth $X_0=x\in\R_+$, and wishes to determine how much to invest in the risky asset.  The investor can purchase the information from an {\IA} that  charges $C$ to provide her with the value of $Y$ before the investor makes decisions. 
To be more precise, at time $0$, the investor chooses $H\in\{0,1\}$ and pays $H\,C$ to the {\IA}. That $H\in\{0,1\}$ is the investor's decision variable corresponding to acquiring the information about $Y$ or not.
Conditional on the information the investor has, she chooses the position $\varphi$ in the stock    
to maximize her exponential utility of the terminal wealth $X_1$, i.e., she finds the maximizer of 
\begin{align}\label{target: single period}
 \max_{\varphi\,\in\, \sigma(H\,Y)-\text{measurable}}\E[\,-\textstyle\exp\left\{-\gamma \,X_1\right\}\,],
\end{align}
where the initial and terminal wealth  are
\begin{align}
\label{dyn: single-period wealth}
X_0 = x-H\,C, \quad X_1 := X_0 + \varphi \,(S_1-S_0) = x+ \varphi\,(\mu+Y+\sigma_ZB_1^Z)- H\,C,
\end{align}
and the strategy $\varphi$ is $\sigma(H\,Y)$-measurable.
When $H=1$, the investor determines her strategy knowing $Y$, hence $\varphi$ is  $\sigma(Y)$-measurable. When $H=0$, we have $H\,Y = 0$, which renders the $\sigma$-algebra trivial,  the investor makes decisions without any information  and her strategy is now $\varphi\in\R$.

\subsection{Single-Period Indifference Price}
\label{ssec: single-period pricing}

We define the indifference price of the extra information as the price $\wC$ such that the investor's expected utility is the same whether or not she purchases the information. To this end, we determine the investor's optimal strategies and her value functions in both the uninformed (UI) and informed (I) cases, and report the indifference price $\wC$ in Theorem~\ref{thm: single-period} -- the proof may be found in  Section~\ref{proof: single-period}. 
That is, $\wC$ satisfies $V^{I}(x,y;\wC) = V^{UI}(x,y)$.

\begin{theorem}
\label{thm: single-period}
In the single-period model setting, the following statements hold:
\begin{enumerate}
\item In the informed case, with purchasing information $Y$ with price $C$ the optimal informed strategy is $\varphi_I^* = ({\mu+Y})/{\gamma\sigma_Z^2}$,
and the optimal expected utility of the investor is
\begin{align}\label{i: single-period utility}
V^I(x,y;C)=
-\textstyle\exp\left\{
-\gamma\,(x-C)-\frac{(\mu+y)^2}{2\left(\sigma^2_Y+\sigma^2_Z\right) }- \frac{1}{2} \log\left(1+ \frac{\sigma^2_Y}{\sigma^2_Z}\right)
\right\}
. 
\end{align}
\item In the uninformed case,  the optimal strategy is $\varphi^*_{UI} = ({\mu+y})/{\gamma\,(\sigma_Y^2+\sigma_Z^2)}$, and the optimal expected utility of the investor is
\begin{align}\label{ui: single-period utility}
V^{UI}(x,y) = -
\textstyle\exp\left\{
-\gamma x-\frac{(\mu+y)^2}{2\left(\sigma^2_Y+\sigma^2_Z\right)}
\right\}.
\end{align}
\item The indifferent price is 
\begin{align}\label{c: single-period}
\wC = \textstyle \frac{1}{2\gamma}\log\left(1+ \frac{\sigma^2_Y}{\sigma^2_Z}\right) \geq 0.
\end{align}
\end{enumerate}
\end{theorem}
  
As long $\sigma_Y>0$, then $\wC>0$. The strict inequality stems from the trading signal $Y$ being only one of the sources of randomness in the market, the other being $B_1^Z$.
We can easily see that the indifference price $\wC$ is increasing in the signal-to-noise ratio ${\sigma_Y}/{\sigma_Z}$. That is, the larger the variance of the trading signal compared to the variance of the stock price, the higher the value of the information.
The other important parameter that appears in the indifference price $\wC$  is the investor's risk aversion $\gamma$. 
To wit, the indifference price $\wC$  is proportional to the risk capacity $1/\gamma$ -- the reciprocal of the investor's risk aversion $\gamma$.


\section{A Tractable Continuous-Time Model}\label{sec: continuous}

In the continuous-time version of the problem, the investor no longer na\"ively chooses a static strategy at time $0$, rather, she can filter the trading signal process $Y$ based on the stock price. We next explicitly show how the information filtering is conducted, and develop an indifference price in the continuous-time setup.

\subsection{Continuous-Time Model Setup}\label{ssec: continuous setup}
For the continuous-time analogue of ~\eqref{dyn: single-period stock},
we consider the following hierarchical structure:
\begin{align}
\label{dyn: continuous-time stock}
dS_t = \left(\mu+Y_t\right) dt +\sigma_z \,dB^Z_t, \qquad S_0 \in\R
\end{align}
where the trading signal process $(Y_t)_\tT$ satisfies the SDE
\begin{align}
\label{dyn: factor process}
dY_t = \sigma_y\, dB^Y_t, \qquad Y_0 = y\in\R.
\end{align}
Here, the expected excess return $\mu>0$, the volatility of the stock $\sigma_z$, and the volatility of the trading signal process $\sigma_y$ are all treated as known constants to the public.
There are two sources of randomness in the stock price dynamics: the Brownian motion $B^Z$ that drives the price volatility, and the trading signal process $Y$, driven by an independent Brownian motion $B^Y$. 
The information regarding the value of the trading signal process $Y$ is not publicly available, but can be purchased through an {\IA} at time $0$ with charge $C(0;T)$. Different from Kyle-type insider trading models~\cite{amendinger1998additional, banek2003information, banerjee2020strategic, brody2009informed, davis1991value, kyle.85}, the {\IA} only reveals the information of $Y$ adaptively, i.e., at time $t$, the {\IA} only knows the historical value $(Y_u)_{u\in[0,t]}$, but not any value of $Y$ at times later than $t$. Thus, we can view the {\IA} as providing precision on the expected instantaneous return of the stock, rather than precision on the value of the stock at maturity.

As in the single-period model in Section~\ref{sec: single period},  there is an investor who participates in the market but does not have access to the trading signal process $Y$. Starting with initial wealth $X_{0-}=x\in\R$ slightly before  time 0, there are two options for the investor: (i) subscribe to the information from the \IA~regarding $Y$ and make an informed decision ($H=1$), or (ii) do not subscribe to the information and make decisions by filtering $Y$ ($H=0$).

If the investor purchases the information about $Y$, at each time $\tT$, she can decide on a trading strategy $\varphi_t$ based on the whole history of the stock price $(S_u)_{u\in[0,t]}$  and the trading signal process $(Y_u)_{u\in[0,t]}$. In this case, by setting $H=1$, the admissible set of $\varphi$, denoted by $\A(S,Y)$, contains all processes adapted to the filtration generated by the stock price and the trading signal process. Contrastingly, if the investor chooses not to purchase the information on the process $Y$, at each time $\tT$, she can only choose a trading strategy $\varphi_t$ based only on the past information of the stock price $(S_u)_{u\in[0,t]}$. In this case, by setting $H=0$, the admissible set of $\varphi$ shrinks to $\A(S)$, which contains  processes adapted to the filtration generated by the stock price alone. In summary, based on the information she has, the investor maximizes her exponential utility of the terminal wealth $X_T$, i.e., she aims to solve for the maximizer of
\begin{align}\label{target: continuous}
 \max_{\varphi\in\A(S, HY)}\E[-\textstyle\exp\left\{-\gamma \,X_T\right\}],
\end{align}
where the wealth process is 
\begin{align}
\label{dyn: continuous-time wealth}
X_t =X_0+ \textstyle\int_0^t \varphi_u\,(\mu+Y_u)\, du + \sigma_z \textstyle\int_0^t \varphi_u\, dB_u^Z, \qquad
X_0 = x - H\,C(0;T). 
\end{align}
  
In this setting, our goal is to determine the indifference price of acquiring the information for the whole interval $\T$, which, similarly to the discrete-time problem in Section~\ref{sec: single period}, requires solving the portfolio optimization problems in both the uninformed and informed cases.

\subsection{Continuous-Time Indifference Price}
Determining the price of information requires solving the portfolio optimization problem with partial information. The extant literature on the topic is vast, e.g.,  see~\cite{al2018outperformance, benevs1983estimation, bjork2010optimal, brody2009informed, casgrain2019trading, cvitanic2006dynamic, detemple1986asset, ferland2008fbsde,fouque2015filtering, guasoni2006asymmetric,
papanicolaou2019backward, pikovsky1996anticipative,rieder2005portfolio,   xia2001learning}, including numerical algorithms such as~\cite{ludkovski2009simulation, ludkovski2012finite, yang2023decision}. Closed form solutions, however, are limited. 

In our setting, as the stock price and the factor process are jointly Gaussian, information filtering (see, e.g.,~\cite{al2018outperformance,  brendle2006portfolio, guasoni2006asymmetric, hitsuda1968representation}) 
can be explicitly constructed. We summarize the key result in Lemma~\ref{lem: filtering}, and refer to Appendix~\ref{appdx: filtering} for details.
\begin{lemma}
\label{lem: filtering}
Let $\F^S:=\{\F_t^{S}\}_{\tT}$ denote the natural filtration generated by the stock price and $(\wY_t)_{\tT}$ denotes the filtered process, s.t., $\wY_t:=\E\left[Y_t\,|\,\F_t^{S}\right]$. Then
\begin{align}
\label{eqn: filtered Y}
\textstyle \wY_t = y+\sigma_y \int_0^t \tanh\left(\frac{\sigma_y}{\sigma_z}\,u\right) d\widehat{B}_u^Z, 
\end{align}
where $\big(\widehat{B}_t^Z\big)_{\tT}$ is an $\F^S$-Brownian motion defined by
\begin{align}
\label{eqn: filtered BM}
\textstyle\widehat{B}^Z_ t = \int_0^t \frac{1}{\sigma_z} \left(Y_u - \wY_u\right) du + B_t^Z.
\end{align}
\end{lemma}

With the information filtering Lemma~\ref{lem: filtering}, particularly the explicit form of $\wY_t$ in~\eqref{eqn: filtered Y}, we can solve the informed and uninformed portfolio optimization problem~\eqref{target: continuous} in closed form. The results are summarized in Theorem~\ref{thm: continuous}.
\begin{theorem}
\label{thm: continuous}
In the continuous-time  setting  in Section~\ref{ssec: continuous setup},  
the following statements hold:
\begin{enumerate}
\item The optimal informed trading strategy is $\varphi^I_t = ({\mu+Y_t})/{\gamma\sigma_z^2}$, and the
optimal wealth $X_t$ follows~\eqref{dyn: continuous-time wealth} with $H=1$. 
The investor's value function is therefore
\begin{align}\label{value:I}
V^I(t,X_t,Y_t; C(0;T))
= -\textstyle\exp\Big\{
&-\gamma X_t-\textstyle\frac{(\mu+Y_t)^2}{2\sigma_y\sigma_z}\tanh\left(\frac{\sigma_y}{\sigma_z}(T-t)\right) 
\notag\\&\quad
-\tfrac{1}{2}\textstyle\log\cosh\left(\frac{\sigma_y}{\sigma_z}(T-t)\right)
\Big\}
. 
\end{align}
\item The optimal uninformed trading strategy (with filtering of the trading signal $\wY_t$) is $\varphi^{UI}_t = (\mu+\wY_t)\cosh\left(\frac{\sigma_y}{\sigma_z}(T-t)\right)\cosh\left(\frac{\sigma_y}{\sigma_z}t\right)\Big/{\gamma\sigma_z^2 \cosh\left(\frac{\sigma_y}{\sigma_z}T\right)}$, and the optimal wealth $X_t$ follows~\eqref{dyn: continuous-time wealth} with $H=0$.
The investor's value function is therefore
\begin{align}\label{value:UI}
V^{UI}(t, X_t, \wY_t) 
= -\textstyle\exp\Bigg\{&-\gamma X_t - \textstyle\frac{(\mu+\wY_t)^2}{2\sigma_y\sigma_z}\frac{\sinh\left(\frac{\sigma_y}{\sigma_z}(T-t) \right)\cosh\left(\frac{\sigma_y}{\sigma_z}t\right)}{\cosh\left(\frac{\sigma_y}{\sigma_z}T\right)}+\frac{1}{2}\log\frac{\cosh\left(\frac{\sigma_y}{\sigma_z}t\right)}{\cosh\left(\frac{\sigma_y}{\sigma_z}T\right)} 
\notag\\&\quad+ \textstyle\frac{\sigma_y}{4\sigma_z}(T-t)\tanh\left(\frac{\sigma_y}{\sigma_z}T\right)+  \frac{\sinh\left(\frac{\sigma_y}{\sigma_z}(T-t)\right)\sinh\left(\frac{\sigma_y}{\sigma_z}t\right)}{4\cosh\left(\frac{\sigma_y}{\sigma_z}T\right)}
\Bigg\}.
\end{align}
\item The indifferent price $\wC(0;T)$ is 
\begin{align}\label{c:continuous}
\wC(0;T) =\textstyle \frac{\sigma_y}{4\gamma\sigma_z}T \tanh\left(\frac{\sigma_y}{\sigma_z}T\right) \geq 0.
\end{align}
\end{enumerate}
\end{theorem}
\begin{proof}
To ease notation, we define
 $h(t) := \textstyle\tanh\left(\frac{\sigma_y}{\sigma_z}t\right)$ and solve the optimization problem in both the informed and uninformed cases.
\paragraph{Informed Case} In the informed case, at time 0, the investor pays $C(0;T)$ to acquire the full information of $Y_t$ on the trading interval $[0,T]$ and, hence, has admissible strategies $\varphi\in\A(S,Y)$. The investor's wealth at time $t$ is thus given by~\eqref{dyn: continuous-time wealth}, and
her goal is to choose $\varphi\in\A(S,Y)$ 
 to maximize \eqref{target: continuous} with $H=1$.
The corresponding value function is defined as $
V^I (t, X_t, Y_t; C(0;T)) := \sup_{\varphi\in\A(S,Y)} \E\left[ - \textstyle\exp\left\{-\gamma X_T\right\} |\mathcal{F}_t\right]
$.
Standard arguments suggest that $V^I(t,x,y; C(0;T))$ satisfies the HJB equation:
\begin{equation}\label{HJB:full}
\frac{\partial V^{I}}{\partial t} + \frac{1}{2}\sigma_y^2  \frac{\partial^2 V^{I}}{\partial y^2} +  \sup_\varphi \left\{ \frac{1}{2}\sigma_z^2\varphi^2 \frac{\partial^2 V^{I}}{\partial x^2} + \varphi (\mu+y)\frac{\partial V^{I}}{\partial x} \right\} =0, 
\end{equation}
with terminal condition $V^{I}(T, x, y; C(0;T)) = - \textstyle\exp\left\{-\gamma (x-C(0;T))\right\}$.
Point-wise optimization yields the optimal controls in feedback form $\varphi^I = -((\mu+ y){\partial V^I/\partial x})/(\sigma_z^2{\partial^2 V^I/\partial^2 x})$.
Making the ansatz that
$
V^I(t,x,y; C(0;T)) =-\textstyle\exp\left\{-\gamma(x-C(0;T)) +A_I(t) (\mu+y)^2 + B_I(t)\right\} 
$, and substituting into the HJB equation~\eqref{HJB:full}, we find
 $(A_I,B_I)$ satisfies the following ODEs:
$$
\textstyle
\tfrac{d}{dt}A_I(t)  = \frac{1}{2\sigma_z^2}-2 \sigma_y^2 A_I(t)^2  \quad \text{and} \quad
\tfrac{d}{dt}B_I(t)  = - \sigma_y^2 A_I(t),
$$
with terminal conditions $A_I(T)=B_I(T)=0$.
Solving this system of ODEs, we obtain that
\begin{align*}
\textstyle
A_I(t) = -\frac{1}{2\sigma_y\sigma_z} \tanh\left(\frac{\sigma_y}{\sigma_z}(T-t)\right), \quad
B_I(t) = - \frac{1}{2}\log \cosh\left(\frac{\sigma_y}{\sigma_z}(T-t)\right) . 
\end{align*}
Accordingly, $\varphi^I_t = (\mu+Y_t)/\gamma\sigma_z^2$, and the value function is~\eqref{value:I}. 
By Corollary 5.14 in\cite[Chapter 3]{karatzas2012brownian}, the value function is a true martingale.  

\paragraph{Uninformed Case}
Now let's assume the investor does not purchase the full information but filters it from $S$. 
By Lemma~\ref{lem: filtering}, 
using the filtered signal process $\wY$ from~\eqref{eqn: filtered Y} and the Brownian motion $\wB$ from~\eqref{eqn: filtered BM} under the restricted filtration $\{\mathcal{F}_t^{S}\}_{\tT}$, we can re-write the stock price dynamic as 
$d S_t 
= (\mu+\wY_t)dt + \sigma_z d\widehat{B}_t^Z
$.
As before, for wealth $X_t$ follows~\eqref{dyn: continuous-time wealth} with $H=0$, the value function associated with the problem \eqref{target: continuous} is defined as
$
V^{UI}(t, X_t, \wY_t) :=  \max_{\varphi\in\A(S)} \E\left[ -\textstyle\exp\left\{-\gamma X_T\right\}| \mathcal{F}_t^S\right].
$
From standard analysis, $V^{UI}(t,x,y)$ satisfies the HJB equation (where we  write $y$ instead of  $\widehat{y}$):
\begin{multline}
\label{HJB:filtered}
  {\textstyle\frac{\partial V^{UI}}{\partial t}}
  \!+\! \sup_\varphi 
  {\textstyle \Big\{ \frac{1}{2}\sigma_z^2\varphi^2 \frac{\partial^2 V^{UI}}{\partial x^2} + \varphi\left(\sigma_y\sigma_zh \frac{\partial^2 V^{UI}}{\partial x \partial y} + (\mu+y)\frac{\partial V^{UI}}{\partial x} \right)\Big\} 
  + \frac{1}{2}\sigma_y^2 h^2 \frac{\partial^2 V^{UI}}{\partial y^2} =0,}   
\end{multline}
subject to the terminal condition $V^{UI}(T, x, y) = - e^{-\gamma x}$.
Point-wise optimization provides the optimal controls in feedback form:
$
\varphi^{UI} = -\textstyle(\sigma_y\sigma_zh{\partial^2 V^{UI}}/{\partial x \partial y} + (\mu+y){\partial V^{UI}}/{\partial x} )\allowbreak
/(\sigma_z^2\partial^2 V^{UI}/\partial^2 x)
$.
Making a similar ansatz as before
$
V^{UI}(t, x, y) = -\exp\{-\gamma x +A_{UI}(t) (\mu+y)^2 +B_{UI}(t)\}
$, and substituting into the HJB equation above,
we find $(A_{UI}, B_{UI})$ satisfies the following ODEs:
$$
\textstyle
 \frac{d}{dt}A_{UI}(t) = \frac{1}{2\sigma_z^2}+\frac{2\sigma_y}{\sigma_z}h(t)A_{UI}(t), \quad\text{and}\quad
 \frac{d}{dt}B_{UI}(t) = -\sigma_y^2 h^2(t)A_{UI}(t) ,
$$
subject to $A_{UI}(T)=B_{UI}(T)=0$.
Solving the ODEs yields
\begin{align}
A_{UI}(t) &=-\textstyle\frac{\sinh\left(\frac{\sigma_y}{\sigma_z}(T-t)\right)\cosh\left(\frac{\sigma_y}{\sigma_z}t\right)}{2\sigma_y\sigma_z \cosh\left(\frac{\sigma_y}{\sigma_z}T\right)}, \label{UI:A}
\\B_{UI}(t) 
&=  \textstyle\frac{\sigma_y}{4\sigma_z}(T-t)\tanh\left(\frac{\sigma_y}{\sigma_z}T\right)+\frac{1}{2}\log\frac{\cosh\left(\frac{\sigma_y}{\sigma_z}t\right)}{\cosh\left(\frac{\sigma_y}{\sigma_z}T\right)} +  \frac{\sinh\left(\frac{\sigma_y}{\sigma_z}(T-t)\right)\sinh\left(\frac{\sigma_y}{\sigma_z}t\right)}{4\cosh\left(\frac{\sigma_y}{\sigma_z}T\right)},\label{UI:B}
\end{align}
and the value function in the uninformed case is therefore~\eqref{value:UI}, 
with the optimal (filtering) strategy given by the expression stated in Theorem~\ref{thm: continuous}.
By Corollary 5.14 in~\cite[Chapter 3]{karatzas2012brownian},
the value function is a true martingale. 
\end{proof}

As in the single-period setting in Section~\ref{ssec: single-period pricing}, 
the indifference price $\wC(0;T)$ is 
proportional to the reciprocal of the investor's risk aversion $\gamma$, and is
monotonically increasing with respect to  ${\sigma_y T}/{\sigma_z}$.
Moreover, if the trading horizon $T$ is large compared to the ratio ${\sigma_z}/{\sigma_y}$, then the indifference price is approximately linear in $T$. 
We term the ratio of the lhs to $T$ as the \emph{limiting} subscription rate and denote it by $\barc$, i.e.,
\begin{equation}
\label{average rate}
\textstyle\barc =  \frac{\sigma_y}{4\gamma\sigma_z}\tanh\left(\frac{\sigma_y}{\sigma_z}T\right).
\end{equation}
We also note that, as long as $\sigma_z>0$, we have the inequality $\textstyle\barc  \leq \frac{1}{4\gamma}\,\frac{\sigma_y}{\sigma_z}$.
This suggests that the \IA~ charges  the ``noise ratio'' $\sigma_y/\sigma_z$,
i.e. how large the volatility $\sigma_y$ of the factor process $Y$ is compared with the volatility $\sigma_z$ of the stock price $S$. 

\section{Best Time to Subscribe}\label{sec:subscribe}

The limiting subscription rate introduced in \eqref{average rate} is the rate the investor is willing to pay starting at time zero until the the end of the trading horizon -- specifically, the decision is made once at the beginning of the trading horizon such that she is indifferent between purchasing the information or filtering the information. 
The investor may, however, wish to wait and purchase the information only at a later point in time. To this end, let $c(t)$ denote the subscription rate set by the {\IA} for receiving information at time $t$. The relationship of the subscription rate $c(t)\geq 0$ and the one-time price $C(0;T)$ in Section~\ref{sec: continuous} is given by
$
C(0;T) = \int_0^T c(t) dt <\infty
$.
Notably, this rate is assumed to be deterministic and does not depend on the stock price $S$ or the trading signal $Y$ -- if it did depend, e.g., on $Y$, then the rate itself  reveals $Y$ and hence the investor would not purchase the information and instead simply query the price of the information. 

With given subscription rate $c(t)$, the investor's optimization problem is to find the optimal time $\tau$ at which to purchase the information, and the admissible trading strategy $\varphi\in\A(S, Y^\tau)$, where the process $Y^\tau$ is defined as $Y_t^\tau = y \one_{\{t\in[0,\tau)\}} + Y_t \one_{\{t\in[\tau,T]\}}$.
That is, (i) for all $t<\tau$, the investor filters $Y$ and optimal trades based on the filtered process; (ii) for all $t\in[\tau,T]$, she is subscribed to the information feed, and receives information about $Y$, and thus has no need to filter it. To define the optimization problem, let $\mathcal{T}$ denote the collection of all $\mathcal{F}^S$-stopping times, the investor aims to find the maximizer of 
\begin{align}
\label{prob: subscription}
\max_{\tau\in\mathcal{T}, \;
\varphi\in\A(S, Y^\tau)}
\E\big[-\textstyle\exp\left\{-\gamma X_T\right\}\big],
\end{align}
where the wealth follows 
\begin{align}
\label{wealth dyn: subscription}
X_t =\textstyle x -  \int_0^t c(u)\one_{[u\ge\tau]}\, du 
+ \int_0^t \varphi_u\,(\mu+Y_u)\, du + \sigma_z \int_0^t \varphi_u\, dB_u^Z.
\end{align}

Next, we prove that the optimal time to subscribe is not necessarily unique and that when it is unique, the optimal time is deterministic. When it is not unique, any time $t$ in the interval $[\tau_e[c],\tau_\ell[c]]$ satisfying $\textstyle\int_t^{\tau_\ell[c]}
\Big( c(s)-\barc+\ell(s)\Big) \; ds= 0$, where $\tau_e[c]$ and $\tau_\ell[c]$ are deterministic, is equally optimal.
The deterministic time $\tau_\ell[c]$ is defined\footnote{Here, if there is no $t$ that satisfies the strict inequality in \eqref{tau: latest}, without loss of generality we set $\tau_\ell[c]=T$.} as
\begin{equation}
\label{tau: latest}
\tau_\ell[c] := \inf\left\{\tT:\textstyle\int_t^u \Big(c(s)-\barc+ \ell(s)\Big) ds <0, \quad \forall u\in(t,T]\right\},
\end{equation}
where $\barc$ 
is the constant subscription rate induced by the indifference price in Theorem \ref{thm: continuous} and given by~\eqref{average rate} and $\ell(t):=\sigma_y\sinh(\frac{\sigma_y}{\sigma_z}(T-2\,t))/4\gamma\,\sigma_z\cosh(\frac{\sigma_y}{\sigma_z}T)$. Further, the deterministic time $\tau_e[c]$ is defined  as
\begin{equation}
\label{eqn: best time}
\tau_e[c]= \inf\left\{\tT:\textstyle\int_t^{\tau_\ell[c]}
\Big( c(s)-\barc+\ell(s)\Big) \; ds= 0\right\}.
\end{equation}
We then have the following Proposition~\ref{prop: subscription}. 
\begin{proposition}
\label{prop: subscription}
Suppose the {\IA} charges the rate $(c(t))_{t\ge0}$, then the optimal time for the investor to start the subscription is any  $t\in[\tau_e[c]\,,\,\tau_\ell[c]]$ such that $\textstyle\int_t^{\tau_\ell[c]}
\Big( c(s)-\barc+\ell(s)\Big) \; ds= 0$. 
\end{proposition}
\begin{proof}
First we investigate the value function once information has been purchased. If the investor purchases the information at time $\tT$, then from the wealth dynamic~\eqref{wealth dyn: subscription} and a similar procedure to derive the informed value function~\eqref{value:I}, we have
\allowdisplaybreaks
\begin{align*}
V_t^I  
= -\textstyle\exp\left\{-\gamma(X_t -\int_t^T c(s)ds )-\frac{\left(\mu+{Y}_{t}\right)^2}{2\sigma_y\sigma_z}\tanh\left(\frac{\sigma_y}{\sigma_z}(T-{t})\right)-\frac{1}{2}\log \cosh\left(\frac{\sigma_y}{\sigma_z}(T-t)\right)\right\}.
\end{align*}
The expectation $\wV^I_t$ of the informed value function just prior to acquiring the information is
\begin{align}\label{value:subscribe}
\wV^I_t :=
\E\left[ V_t^I|\mathcal{F}^S_t\right] 
= V^{UI}(t, X_t, \wY_t)\,\textstyle\exp\Big\{\gamma \int_t^T c(s) ds-B_{UI}(t) + \tfrac{1}{2}\log \frac{\cosh\left(\frac{\sigma_y}{\sigma_z}t\right)}{\cosh\left(\frac{\sigma_y}{\sigma_z}T\right)}\Big\}. 
\end{align}
The second equality follows from using $V^{UI}$  in~\eqref{value:UI} and deterministic $B_{UI}$ in~\eqref{UI:B}.

Next, we show that $\tau_\ell[c]$ in \eqref{tau: latest} is an upper bound on the time at which the investor is indifferent to purchasing the information feed. To this end, if the investor  purchases information at $u\in[t,T]$, then she uses an uninformed strategy on $[t,u]$ (with filtering to obtain $\wY$) and uses an informed strategy on $[u,T]$ (with knowledge of $Y$). 
As $(V^{UI}_t)_{\tT}$ is a true $\mathcal{F}^S$-martingale (by~\eqref{value:UI}),  we have
\begin{align}
 \E[\wV_u^{I}|\mathcal{F}^S_t] 
&=V^{UI}_t\;\textstyle\exp\Big\{\gamma \int_u^T c(s) ds-B_{UI}(u) + \frac{1}{2}\log \frac{\cosh\left(\frac{\sigma_y}{\sigma_z}u\right)}{\cosh\left(\frac{\sigma_y}{\sigma_z}T\right)}\Big\}. \label{condition: u} 
\end{align}
The investor will not purchase the information later than $t$ if  $\E[\wV_u^{I}|\mathcal{F}^S_t]<\wV_t^I$, $\forall\;u\in(t,T]$. 
Comparing~\eqref{value:subscribe} and~\eqref{condition: u} and plugging in the expression for $B_{UI}$ from~\eqref{UI:B}, this means
\begin{align}\label{ineq:latest}
\textstyle\int_t^u c(s)\, ds 
&<
\textstyle\int_t^u \barc-\ell(s) ds. 
\end{align}
Thus, with $\barc$ defined in~\eqref{average rate}, we introduce the `latest subscription time' $\tau_\ell[c]$ of a subscription rate $c$ in \eqref{tau: latest}.
If $\tau_{\ell}[c]\in[0,T]$, suppose the investor decides to purchase the information at time $\tau_0$, where $\tau_0$ is a stopping time in $\mathcal{T}$, we can infer the \emph{strict} inequality~\eqref{ineq:latest} holds if $\tau_0>\tau_\ell[c]$. Therefore, if $\tau_0>\tau_\ell[c]$, 
then $
\E\left[\wV_{\tau_0}^I|\mathcal{F}^S_{\tau_\ell[c]}\right] < \wV^I_{\tau_\ell[c]}
$, the investor will immediately purchase the information at $\tau_\ell[c]$. 
By the definition of $\tau_\ell[c]$, there exists $t\in[0,\tau_\ell[c]]$ such that
$$
\textstyle
\int_t^{\tau_\ell[c]} \Big(c(u)-\barc+ \ell(s) \Big)\,du \geq 0, 
$$
which makes $\tau_e[c]$ given in~\eqref{eqn: best time} well defined. 

As we have demonstrated that $\tau_\ell[c]$ is an upper bound on the optimal stopping time, we focus now on the shortened interval $[0,\tau_\ell[c]]$. 
Let $V^F(t, X_t, \wY_t) $ denote the investor's value function   who can choose the purchasing time on $[0,\tau_\ell[c]]$ for the information.
In the continuation region, $V^F$ satisfies
\begin{align*}
\textstyle\sup_\varphi \left\{ \frac{1}{2}\sigma_z^2\varphi^2 \frac{\partial^2 V^{F}}{\partial x^2} + \varphi\left(\sigma_y\sigma_z h \frac{\partial^2 V^{F}}{\partial x \partial y} + (\mu+y)\frac{\partial V^{F}}{\partial x} \right)\right\}
+ \frac{\partial V^{F}}{\partial t} + \frac{1}{2}\sigma_y^2 h^2 \frac{\partial^2 V^{F}}{\partial y^2} =0, 
\end{align*}
Point-wise optimization yields
the optimal controls in feedback form in the continuation region
$
\varphi^{F} = -(\sigma_y\sigma_zh{\partial^2 V^{F}}/{\partial x \partial y} + (\mu+y){\partial V^{F}}/{\partial x})/(\sigma_z^2\partial^2 V^{F}/\partial^2 x).
$
Standard arguments imply that $V^F$ satisfies the quasi-variational inequality on $[0,\tau_\ell[c]]$:
\begin{equation}\label{HJB:flexible time}
\left\{
\begin{aligned}
&
\textstyle\max\Big\{\frac{\partial V^{F}}{\partial t} + \frac{\sigma_y^2}{2} h^2 \frac{\partial^2 V^{F}}{\partial y^2}  
- \frac{\left(\frac{\sigma_y}{\sigma_z}h \frac{\partial^2 V^{F}}{\partial x \partial y} + \frac{\mu+y}{\sigma_z^2}\frac{\partial V^{F}}{\partial x} \right)^2}{2\frac{\partial^2 V^{F}}{\partial x^2}};
\wV^I(t,x,y)- V^{F}(t, x, y)
\Big\} =0\;,
\\
& V^F(\tau_\ell[c],x,y) = \wV^I(\tau_\ell[c],x,y).
\end{aligned}
\right.
\end{equation}
When $V^F(t,x,y) = \wV^I(t,x,y)$, it is optimal for the investor to purchase the information. 
We make the ansatz $V^F(t,x,y) = -\textstyle\exp\left\{-\gamma x +A_{UI}(t)(\mu+y)^2 + B_F(t)\right\}$, where $A_{UI}$ is given by~\eqref{UI:A}. Substituting this ansatz into the quasi-variational inequality~\eqref{HJB:flexible time}, 
we obtain a quasi-variational inequality for $B_F$ on $[0,\tau_\ell[c]]$:
\begin{equation*}
\left\{
\begin{aligned}
&\textstyle\max\Bigg\{
B_F(t)- \textstyle\gamma \int_t^T c(u)du - \frac{1}{2} \log \frac{\cosh\left(\frac{\sigma_y}{\sigma_z}t\right)}{\cos\left(\frac{\sigma_y}{\sigma_z}
T\right)}\;;\;
\frac{d}{dt}B_F(t) +\sigma_y^2\tanh^2\left(\frac{\sigma_y}{\sigma_z}t\right) A_{UI}(t)\Bigg\} =0,
\\
&B_F(\tau_\ell[c])= \gamma \textstyle\int_{\tau_\ell[c]}^T c(u) du + \frac{1}{2} \log \frac{\cosh\left(\frac{\sigma_y}{\sigma_z}\tau_\ell[c]\right)}{\cos\left(\frac{\sigma_y}{\sigma_z}
T\right)}.
\end{aligned}
\right.
\end{equation*}
With $A_{UI}$ by~\eqref{UI:A}, using a comparison argument (cf., e.g.,~\cite{chicone2006ordinary}), the unique solution to this variational ODE is
\begin{align}\label{eq:BF}
B_F(t)
&=\gamma \textstyle\int_{\tau_\ell[c]}^T c(u)\,du + \gamma\int_t^{\tau_\ell[c]} \left(\barc-\ell(u)\right)\, du
+\frac{1}{2}\log\frac{\cosh\left(\frac{\sigma_y}{\sigma_z}t\right)}{\cosh\left(\frac{\sigma_y}{\sigma_z}T\right)}.
\end{align}
It is optimal for the investor to make the purchase at time $\tau\in[0,\tau_\ell[c]]$ whenever $V^F(\tau,x,y) = \wV^I(\tau,x,y)$. Using the above expression and the form of $V^F$ and $V^I$, $\tau$ is deterministic 
\begin{align}
\label{stopping rule}
\textstyle\int_\tau^{\tau_\ell[c]} (c(s) - \barc+\ell(s))\ ds = 0.
\end{align}
There may be multiple $\tau$ satisfying~\eqref{stopping rule}, i.e., at which the investor is indifferent in entering into the subscription. Among them, the earliest time $\tau_e[c]$ is given by~\eqref{eqn: best time}. 
\end{proof}

We next point out two special cases.

\paragraph{I} If the {\IA} charges the constant limiting subscription rate $\barc$ given by~\eqref{average rate}, then $\tau_e[\barc] = \tau_\ell[\barc] = T/2$, and the optimal time is unique. As $\sinh(x)$ is strictly positive on $(0,\infty)$ and strictly negative on $(-\infty,0)$, 
$
\tau_\ell[\barc] = \textstyle\inf\left\{\tT: \int_t^u  \ell(s) ds <0, \forall u\in(t,T]\right\}
= \frac{T}{2}.
$
Notice that for $t\in[0,\tau_\ell[\barc])$, 
$\textstyle\int_t^{\tau_\ell[\barc]} \ell(s)\; ds> 0
$. 
Hence for the limiting subscription rate $\barc$ in~\eqref{average rate},  the investor purchases the information subscription at time $\tau_e[\barc]={T}/{2}=\tau_{\ell}[\barc]$.

\paragraph{II} The indifference subscription rate $\wc$ with $\tau_e[\wc]=0 < \tau_\ell[\wc]=T$ is 
\begin{align}
\label{fair rate finite}
\wc(t) = \barc -\ell(t).
\end{align}
This follows, as the term under the integral in the expression for $\tau_\ell[\wc]$ in \eqref{tau: latest} is identically zero, hence $\tau_\ell[\wc]=T$.
Similarly, by checking~\eqref{eqn: best time} we have $\tau_e[\wc]=0$. 
Therefore, the investor is indifferent between purchasing the information subscription or trading  with the information of the stock over the whole time period, hence the earliest time to make the purchase is therefore $\tau_e[\wc] =0\ne \tau_\ell[\wc]=T$.

\section{Example}
\begin{wrapfigure}{r}{0.5\textwidth}
\centering
\includegraphics[width=0.38\textwidth]{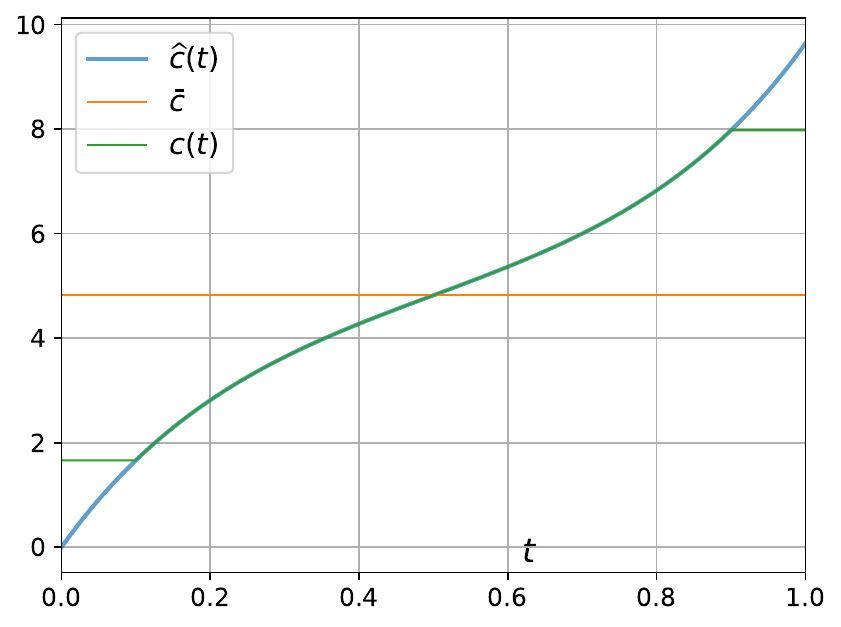}
\caption{\footnotesize{The indifference subscription rate $\wc(t)$, the limiting subscription rate $\overline{c}(t)$, and the prescribed subscription rate $c(t)$.}}
    \label{fig:subscription-rate}
~\\[0.5em]

\includegraphics[width=0.4\textwidth]{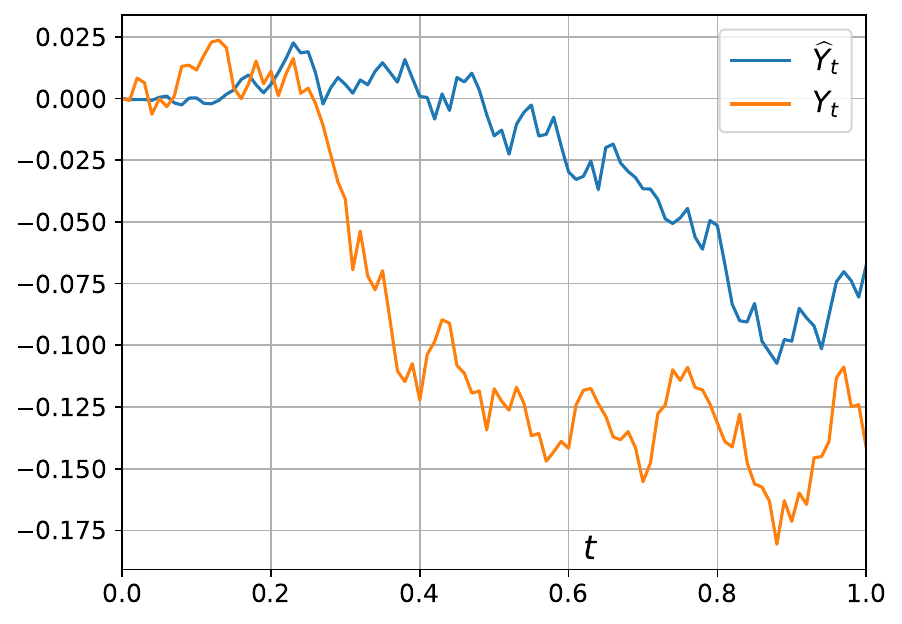}
\caption{\footnotesize{The signal process $Y_t$ and the filtered signal process $\wY_t$.}}
\label{fig:filtered-Y}
~\\[0.5em]
            
\includegraphics[width=0.4\textwidth]{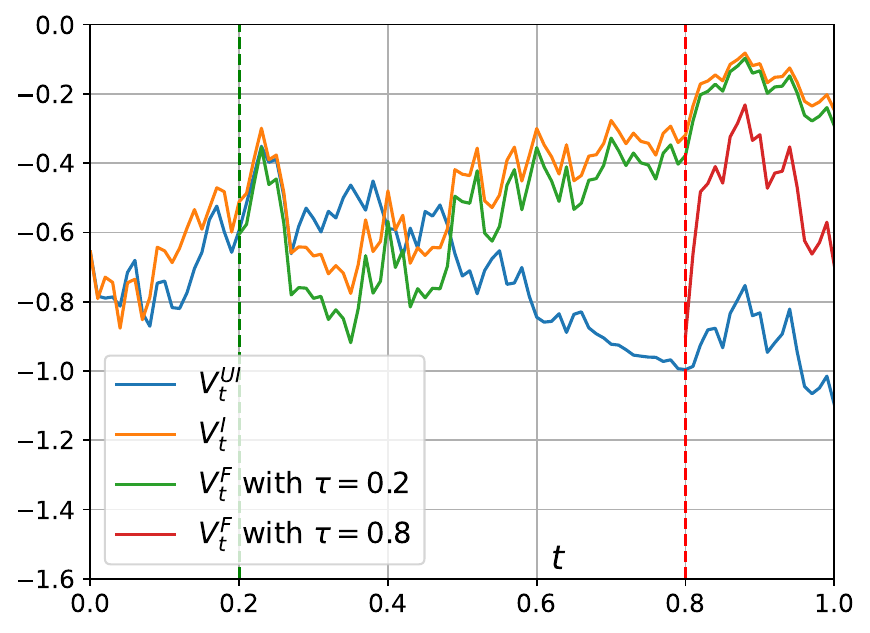}
\caption{\footnotesize{The value functions without information $V_t^{UI}$, with information $V_t^I$ (purchased at time 0 with $\wC(0,T)$) , with prescribed subscription rate $c(t)$ $V_t^F$. }}
~\vspace*{-3em}
\label{fig:value-func}
\end{wrapfigure}  
We provide an illustrative example of the ideas that we develop with the following model parameter assumptions:
\begin{align*}
&\mu     = 0.05, \; 
\sigma_y = 0.1, \;
\sigma_z = 0.05, \;
T=1, \; 
\\
&\gamma = 0.1, \;
x = 0, \;
y = 0, \;
S_0 = 10.
\end{align*}
Using these parameters we obtain the indifference subscription rate $\wc(t)$, and further include the limiting subscription rate $\barc$ and introduce a prescribed subscription rate $c(t)$ shown in Figure \ref{fig:subscription-rate}. As the prescribed subscription rate lies above the indifference rate prior to $t=0.2=\tau_e[c]$, it is never optimal for the investor to subscribe early than $t=0.2$. Similarly, as the prescribed subscription rate lies below the indifference price after $t=0.8=\tau_\ell[c]$, it is optimal for the investor to subscribe no later than $t=0.8$.   
The investor can subscribe at any time between $\tau_e[c]$ and $\tau_\ell[c]$ with equal value. In Figure \ref{fig:filtered-Y}, we show a sample path of the true and  filtered signals $Y$ and $\wY$, respectively. In tandem, Figure \ref{fig:value-func} shows the corresponding sample paths of the value functions with and without subscribing at different times.


\section{Conclusions}

In this study, we examined the valuation and optimal timing of acquiring additional information about an asset's price trajectory. There are several avenues for future research, including  exploring multi-agency information purchases, competitive dynamics among information providers, investigating alternate information flow models (e.g., using mean-reverting processes), and even extending to infinite-horizon subscribing models. Despite the specific model and the underlying dynamics, we can anticipate the price of information is increasing with respect to the signal-to-noise ratio, decreasing with respect to the investor's risk aversion, and proportional to the trading horizon $T$ when it is relatively large.

\section*{Acknowledgments}
We would like to acknowledge  fruitful discussions with Paolo Guasoni, Mike Ludkovski, Steven Shreve, and Zhi Li. We thank Liam Welsh for his careful reading of the first version.
\vspace*{-1em}

\bibliographystyle{siamplain}
\bibliography{references}
\appendix
\section{Proofs}\label{appdix: proofs}
\subsection{Proof of Theorem~\ref{thm: single-period}}
\label{proof: single-period}
\begin{proof}

We solve the optimization problem with  and without purchasing the information. 

\paragraph{Informed case}
In the informed case, the investor chooses to purchase the information, hence over $\sigma(Y)$-measurable $\varphi$, she optimizes
$
\E\left[-\textstyle\exp\left\{-\gamma (x-C+\varphi (Y+Z)) \right\}\;|\;\sigma(Y)\right]
= -\textstyle\exp\left\{-\gamma (x-C) -\gamma \varphi(Y+\mu) +\frac{1}{2} \gamma^2\varphi^2 \sigma_Z^2\right\}
$
, which yields
$\varphi_{I}^* = ({\mu+Y})/{\gamma\sigma^2_Z}$
and by the characteristic function of $Y$, we have~\eqref{i: single-period utility}.

\paragraph{Uninformed case}
In the uninformed case, the investor does not purchase the information, and hence over deterministic $\varphi\in\R$ she 
maximizes
$
\E\left[-\textstyle\exp\left\{-\gamma(x+\varphi(Y+Z))\right\}\right] 
=\textstyle\exp\left\{-\gamma x - \gamma\varphi(\mu+y) + \frac{1}{2}\gamma^2\varphi^2\left(\sigma^2_Y+\sigma_Z^2\right)\right\}
$.
Therefore, the (deterministic) optimal strategy is
$
\varphi_{UI}^* = {\mu+y}/{\gamma\left(\sigma^2_Y+\sigma_Z^2\right)},
$
and we have~\eqref{ui: single-period utility}. 

The investor is indifferent to the two options whenever $V^I(x,y;\wC) = V^{UI}(x,y)$ and solving for $\wC$ provides us with~\eqref{c: single-period}.
\end{proof}


%
%



\subsection{Proof of Lemma~\ref{lem: filtering}}\label{appdx: filtering}
\begin{proof}
Here,
we use $L^2([0,T]^2)$ to denote the separable Hilbert spaces of real-valued, square-integrable functions; two functions are considered equal if they coincide almost everywhere.
Using Fokker–Planck equation, we can verify that $(S_t, Y_t, B^Z_t)$ are jointly Gaussian processes. In particular, $S_t - (\mu+y)t-S_0$ is a mean zero Gaussian process with covariance function
$
\E\left[S_t , S_u \right] = \sigma_z^2 u +\sigma_y^2 \left(\frac{t}{2}-\frac{u}{6}\right)u^2 = \sigma_z^2 u + \sigma_y^2\textstyle\int_0^t \int_0^u \min\{v,s\} dv ds 
$,
$\forall\;0\leq u\leq t\leq T$.
By checking the conditions in~\cite[Theorem 1]{shepp1966radon}, and by~\cite[Proposition 2]{hitsuda1968representation}, 
there exists a unique Brownian motion $(\widehat{B}^Z_t)_{\tT}$ and a unique kernel $\kappa\in L^2([0,T]^2)$ such that
\begin{align}\label{eqn: hBz}
S_t = S_0 +(\mu+y)t + \sigma_z \,\widehat{B}^Z_t -\textstyle\int_0^t \int_0^u \kappa(u,v) \,d\widehat{B}^Z_v, \qquad 0\leq t\leq T.
\end{align}
Here, the kernel $\kappa$ satisfies the following integral equation:
$$
\sigma_z\,\kappa(t,u) - \textstyle\int_0^u\kappa(t,v)\,\kappa(u,v)\,dv = -{\sigma_y^2} u, \qquad 0\leq u\leq t\leq T,
$$
and following similar arguments as in~\cite[Section 5]{guasoni2006asymmetric}, $\kappa$ can be calculated explicitly as
\begin{align}
\kappa(t,u) = - {\sigma_y} h(u) \1{u\in[0, t]}. 
\end{align}
Recall that for the filtration $\{\F_t^{S}\}_{\tT}$ generated by the stock price, we define the projection
$
\wY_t = \E\left[Y_t | \F^S_t\right]. 
$
Hence, by the uniqueness of $(\widehat{B}^Z_t)_{\tT}$ supported by  $\{\F_t^{S}\}_{\tT}$, comparison of the dynamic of $S_t$ from~\eqref{dyn: continuous-time stock} and~\eqref{eqn: hBz} yields
~\eqref{eqn: filtered Y} and~\eqref{eqn: filtered BM} in Lemma~\ref{lem: filtering}.
\end{proof}



\end{document}